\documentclass[lettersize,journal]{IEEEtran}
\usepackage{amsmath,amsfonts}
\usepackage{amsthm}
\usepackage{algorithmic}
\usepackage{algorithm}
\usepackage{array}
\usepackage[caption=false,font=normalsize,labelfont=sf,textfont=sf]{subfig}
\usepackage{textcomp}
\usepackage{stfloats}
\usepackage{url}
\usepackage{verbatim}
\usepackage{graphicx}
\usepackage{svg}
\usepackage{cite}
\usepackage{orcidlink}
\usepackage{multirow}
\usepackage[utf8]{inputenc}
\newtheorem{theorem}{Theorem}
\newtheorem{corollary}{Corollary}
\newtheorem{remark}{Remark}
\renewenvironment{proof}[1][Proof]{\par\noindent\textbf{#1. }\rmfamily}{\qed\par}

\begin{document}

\title{NeuroDOB: A Deep Neural Observer-Based Controller for Vehicle Lateral Dynamics}

\author{\vskip 1em
	Sangmin Kim,
    Taehun Kim,
    Guntae Kim,
	and Chang~Mook~Kang
       } 

\markboth{}{}%

\maketitle

\begin{abstract}
This paper proposes NeuroDOB, a deep neural network based observer controller for vehicle lateral dynamics, which replaces the conventional disturbance observer (DOB) with a deep neural network (DNN) to enhance personalized lateral control. Unlike conventional DOBs that compensate for general disturbances such as road friction variation and crosswind, NeuroDOB explicitly addresses unmodeled vehicle dynamics and driver-specific behaviors by learning the steering compensation signal from driver-in-the-loop simulations using CarSim’s embedded controller as a surrogate driver. The proposed architecture integrates NeuroDOB with a linear quadratic regulator (LQR), where the DNN outputs a delta error correction added to the baseline LQR steering input to produce the final control command. Input features to the DNN include lateral position and yaw angle errors, LQR control input. Experimental validation using a lateral dynamic bicycle model within CarSim demonstrates that NeuroDOB effectively adapts to individual driving habits, improving lateral control performance beyond what conventional LQR controller's achieve. The results indicate the potential of deep neural network based observer to enable personalized and adaptive autonomous vehicle control. 
In cognitive terms, the proposed architecture can be viewed as a dual-system control structure. The baseline LQR corresponds to System 1, a model-based, fast, and analytic reasoning layer ensuring stability. The NeuroDOB acts as System 2, a reflective, data-driven layer that learns compensation from experience and corrects the analytical bias of System 1. Together, they form an integrated decision process analogous to human intuition–reflection interaction, enabling both stability and adaptability in lateral control.
\end{abstract}

\begin{IEEEkeywords}
Vehicle lateral dynamics, Disturbance observer, Deep neural network, Linear quadratic regulator
\end{IEEEkeywords}

\section{Introduction}
\IEEEPARstart{A}{ccurate} and robust lateral control of vehicles is a fundamental factor underpinning the reliability of autonomous driving and advanced driver-assistance systems (ADAS)\cite{ref1}. As the automotive industry progressively advances towards higher levels of driving automation, the development of lateral control algorithms capable of maintaining lanes, tracking paths smoothly, and avoiding collisions with precision and stability has become not only a technical requirement, but also a decisive influence on real-world adoption and user trust in intelligent vehicles. Model-based control methodologies have been widely adopted in both research and industrial vehicle development, owing to their mathematical rigor, ease of implementation, and superior performance when the model accurately reflects the real environment\cite{ref2, ref3}.

Vehicle lateral controllers must contend with a variety of dynamic factors, including tire slip and deformation, heterogeneous road surfaces, sudden steering or load changes, and uncertainties in mass and moment of inertia\cite{ref4, ref5}. The LQR framework is a powerful tool for optimizing these control objectives. However, it is fundamentally based on the assumption that vehicle model parameters are precise and remain constant over time. In reality, vehicle dynamics change significantly due to environmental conditions, driving style, and hardware aging, frequently invalidating the assumptions underlying model-based design\cite{ref6}.

To address these limitations, DOB compensation loops have been combined with LQR controllers\cite{ref7, ref8}. The chief objective of the DOB is to compensate for predictable and measurable disturbances, thereby extending the controller’s robustness within certain bounds. However, DOB performance also relies heavily on model accuracy and delivers optimal results only when disturbances are static or repetitive\cite{ref9}.

In real-world scenarios, unexpected nonlinear dynamics, tire slip, environmental variations, and most importantly idiosyncratic driver behaviors that are difficult to capture in a model are primary causes of control performance degradation. The challenges of unmodeled dynamics and personalization cannot be fundamentally overcome by conventional DOB-LQR structures\cite{ref10, ref11, ref12}.

During mass production of commercial vehicles, controller parameters are typically calibrated to an average driver profile, meaning that diverse driver behavior is largely ignored. Consequently, some users perceive the response as “sluggish,” others as “unstable,” and overall user distrust is fostered. Individual tuning of controllers for every driver is infeasible, and thus, the personalization of vehicle lateral control remains an open research challenge\cite{ref13}.

In response to these issues, research on personalization in vehicle control has recently become more active. Approaches such as controller adaptation based on steering pattern data and driver clustering have been proposed, but these still suffer from scalability limitations and do not perform well for high-dimensional, nonlinear compensation in real-world applications\cite{ref14, ref15}.

DNN-based machine learning approaches have emerged as a promising alternative to overcome these challenges. DNNs are powerful tools capable of extracting patterns from high-dimensional signals and have recently outperformed traditional estimation and classification techniques in diverse driving, monitoring, and prediction domains\cite{ref16}. However, most DNN applications have been limited to perception, event prediction, or auxiliary signal estimation, with very few examples of DNNs directly integrated into observer-compensation structures, especially for real-time driver-style adaptation.

To address these studies, this paper proposes NeuroDOB architecture. The core innovation is the integration of the stability and interpretability of LQR with data-driven, driver-specific DNN compensation. The DNN is designed and trained to fully replace the conventional DOB, learning driver-specific compensation signals. Training and validation are conducted using data from the CarSim embedded controller, treated as a surrogate for real drivers. The input to the DNN includes vehicle states, LQR input, and driver steering angle, while the output is the optimal compensation $\delta_{\mathrm{c}}$.

Through this architecture, the DNN learns and reproduces real-time, driver-specific steering patterns that conventional observers or rule-based adaptation cannot capture. As such, the controller not only compensates for modeling errors and disturbances, but also generates steering compensation signals that follow the behavioral signature of each driver, enabling true “driver-vehicle-controller” cooperation in real time.

The validity of NeuroDOB is demonstrated through comprehensive experimental scenarios employing a lateral dynamic bicycle model within the CarSim environment. Furthermore, additional experimental analyses utilizing real-world driving data from production vehicles are performed, thereby verifying that the proposed methodology can be effectively deployed in practical automotive applications.

The main contributions of this paper are as follows:
\begin{itemize}
    \item We systematically reveal the limitations of existing LQR and DOB-based lateral control in adapting to driver behavior and unmodeled dynamics, and propose a fundamentally new architecture to overcome these challenges.
    \item We develop and experimentally verify NeuroDOB structure centered on real-time driver-specific adaptation.
    \item We provide comprehensive experimental results demonstrating that NeuroDOB delivers superior tracking accuracy and personalized adaptation across different drivers and scenarios, and discuss strategies for practical implementation.
\end{itemize}

\section{System Architecture and Modeling}

Fig.~\ref{fig:NeuroDOB framework architecture} illustrates the overall NeuroDOB framework, which integrates a traditional LQR controller with a NeuroDOB for driver-personalized lateral control in the CarSim environment.
At each time step, the vehicle's lateral and yaw tracking errors and their derivatives, together with the baseline LQR steering command, are provided as inputs to the NeuroDOB module.
The NeuroDOB processes these signals and generates a compensation term $\delta_{\mathrm{c}}$, which reflects driver-specific steering behavior and unmodeled dynamics.
This compensation is algebraically added to the LQR command, resulting in the final steering command $\delta_{\mathrm{f}}$ that is applied to the vehicle.

This architecture enables closed-loop error correction and real-time adaptation to both modeling uncertainties and individual driver preferences.
By directly learning from integrated simulation or real driver data, the NeuroDOB augments the stability and interpretability of the LQR structure with highly flexible data-driven compensation.

\begin{figure}[!ht]
    \centering
    \includegraphics[width=0.45\textwidth]{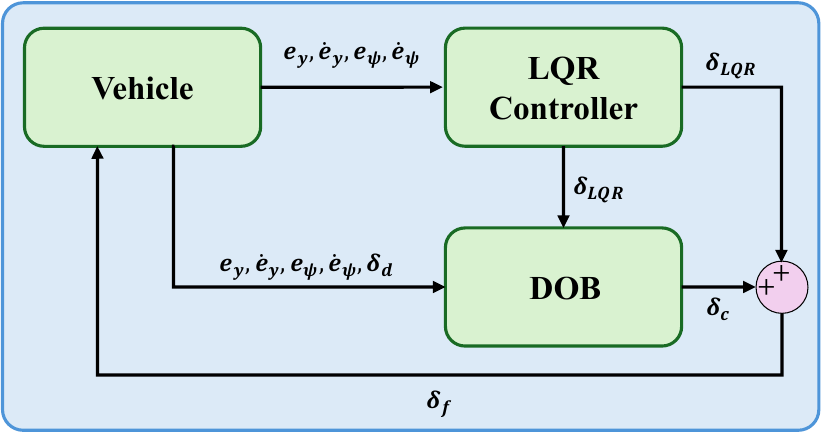}
    \caption{Basic LQR and DOB system architecture.}
    \label{fig:NeuroDOB framework architecture}
\end{figure}

The details of the neural network structure, training procedure, and real-time compensation mechanism are thoroughly described in the following sections.

\subsection{Lateral Bicycle Model}

To describe the vehicle lateral dynamics, a linearized bicycle model, widely adopted in vehicle control literature \cite{ref17, ref18}, is employed.  
The state vector is defined as
\begin{equation}
    x = \left[e_y, \dot{e}_y, e_{\psi}, \dot{e}_{\psi}\right]^{\top}
\end{equation}
where $e_y$ denotes the lateral offset error from the lane center at the vehicle’s center of gravity (CG), $\dot{e}_y$ is its time derivative, $e_{\psi}$ is the yaw angle error, and $\dot{e}_{\psi}$ represents the yaw rate error.

The continuous-time lateral error dynamics of the vehicle are represented in the state-space form as
\begin{equation} \label{eq:cont_sys}
    \dot{x} = A x + B u + B_2 \dot{\psi}_{des}
\end{equation}
where $u$ is the steering input, and $\dot{\psi}_{des}$ denotes the desired yaw rate feedforward term.  
The system matrices $A$, $B$, and $B_2$ characterize the continuous-time vehicle lateral dynamics.  
The continuous-time model is discretized with a sampling time $T_s$ using the forward Euler method~\cite{ref19}, yielding the discrete-time matrices $\Phi$, $\Gamma$ and $\Gamma_2$. This provides a simple yet effective approximation of the continuous-time system behavior within each sampling interval and forms the discrete-time formulation used in the stability analysis of Theorem~1.

\begin{equation}
A = 
\begin{bmatrix}
0 & 1 & 0 & 0 \\
0 & a_{22} & a_{23} & a_{24} \\
0 & 0 & 0 & 1 \\
0 & a_{42} & a_{43} & a_{44}
\end{bmatrix},
\quad
B = 
\begin{bmatrix}
0 \\ b_{21} \\ 0 \\ b_{41}
\end{bmatrix},
\quad
B_2 = 
\begin{bmatrix}
0 \\ b_{2,21} \\ 0 \\ b_{2,41}
\end{bmatrix}
\end{equation}

where the matrix elements are given as
\[
\begin{aligned}
a_{22} &= -\tfrac{2\,(C_{af}+C_{ar})}{m\,V_x}, 
    &\quad a_{23} &= \tfrac{2\,(C_{af}+C_{ar})}{m}, \\[6pt]
a_{24} &= -\tfrac{2\,(C_{af}l_f - C_{ar}l_r)}{m\,V_x},
    &\quad a_{42} &= -\tfrac{2\,(C_{af}l_f - C_{ar}l_r)}{I_z\,V_x}, \\[6pt]
a_{43} &= \tfrac{2\,(C_{af}l_f - C_{ar}l_r)}{I_z},
    &\quad a_{44} &= -\tfrac{2\,(C_{af}l_f^2 + C_{ar}l_r^2)}{I_z\,V_x}, \\[6pt]
b_{21} &= \tfrac{2\,C_{af}}{m},
    &\quad b_{41} &= \tfrac{2\,C_{af}l_f}{I_z}, \\[6pt]
b_{2,21} &= -\tfrac{2(C_{af}l_f - C_{ar}l_r)}{m\,V_x} - V_x,
    &\quad b_{2,41} &= -\tfrac{2(C_{af}l_f^2 + C_{ar}l_r^2)}{I_z\,V_x}.
\end{aligned}
\]

Here, $m$ is the vehicle mass, $I_{z}$ denotes the yaw moment of inertia, $l_f$ and $l_r$ are the distances from the vehicle’s center of gravity to the front and rear axles, respectively, $C_{af}$ and $C_{ar}$ are the cornering stiffness coefficients of the front and rear tires, and $V_x$ represents the constant longitudinal velocity of the vehicle.

\begin{figure}[!ht]
    \centering
    \includegraphics[width=0.35\textwidth]{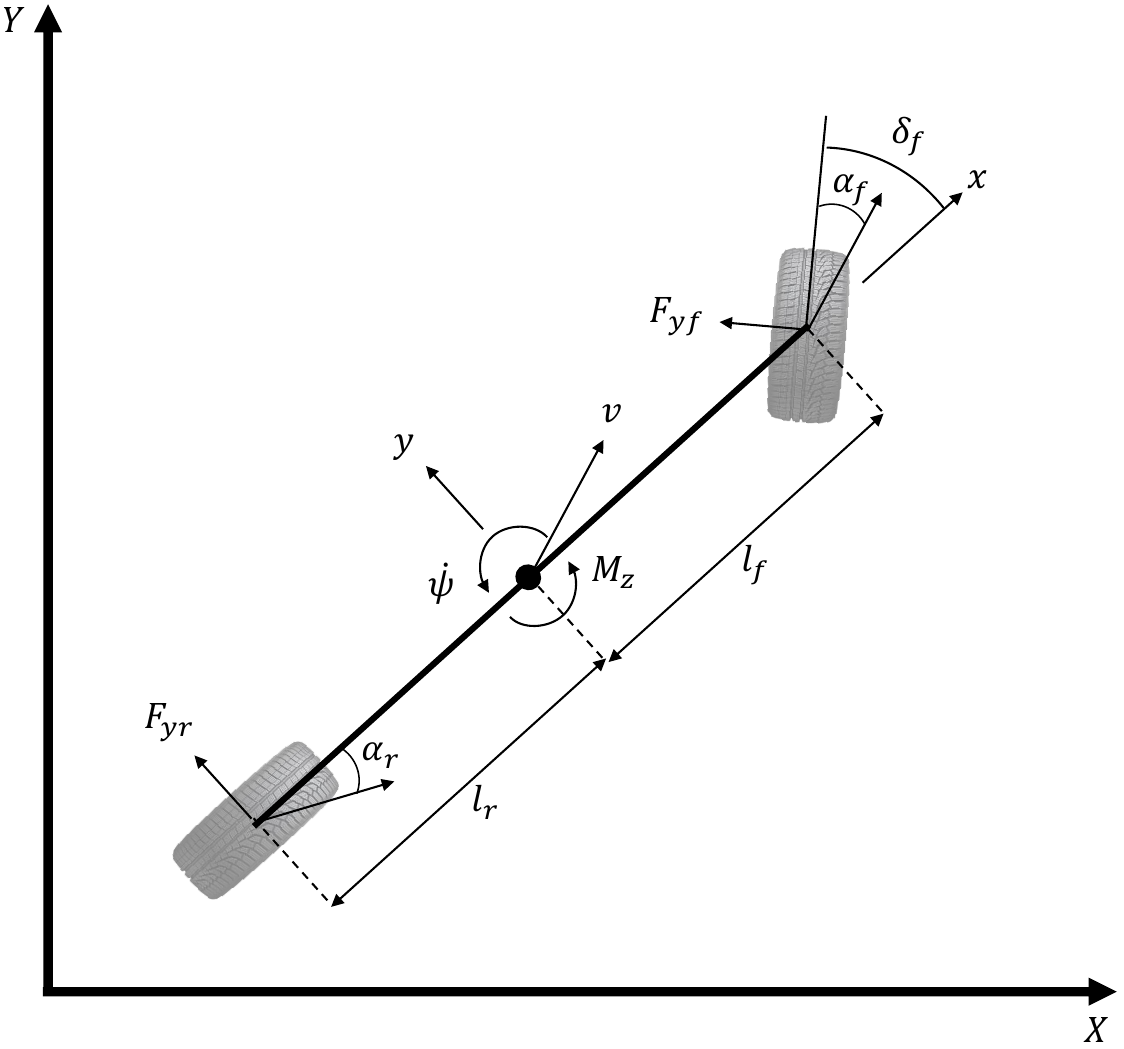}
    \caption{Lateral dynamic bicycle model.}
    \label{fig:bicycle model}
\end{figure}

\subsection{LQR Controller Design}
Based on the state-space model described in the previous subsection, a standard Linear Quadratic Regulator (LQR) controller is designed for lateral trajectory tracking. The LQR approach is widely used in vehicle control due to its optimality and ease of implementation, and it serves as the baseline controller in this study\cite{ref20}. The objective of the LQR controller is to minimize the vehicle's lateral position and yaw angle errors while simultaneously ensuring system stability and efficient use of control inputs\cite{ref21}. The cost function of the LQR controller is constructed as follows:
\begin{equation}
    J = \int_{0}^{\infty} \left( X_k^{T} Q X_k + U_k^{T} R U_k \right) \, dk
    \label{eq:lqr_cost}
\end{equation}
where \( X_k \) is the state variable of the system, \( U_k \) is the control variable of the system, and \( Q \) and \( R \) represent the state error weighting matrix and the control weighting matrix of the controller, respectively. Here, \( Q \) is a positive definite or semi-positive definite matrix, and \( R \) is a positive definite matrix. 
Within the overall NeuroDOB architecture, the LQR controller provides the foundational control performance, which is enhanced by the data-driven compensation signals for improved robustness and driver-specific adaptation.

\subsection{Problem Formulation}

While the LQR controller effectively regulates lateral error, its performance degrades in the presence of unmodeled dynamics, parameter uncertainties, and driver-specific behavioral variations. Conventional DOB partially compensate for external disturbances such as road friction irregularities or crosswinds but fail to adapt to complex and individualized driving patterns.

Our objective is to develop an observer-based compensation scheme that can learn personalized correction signals to augment the baseline LQR input. Specifically, NeuroDOB will replace the conventional DOB, trained to map measured error states and control inputs, along with surrogate driver steering signals, to a delta error compensation $\delta_{\mathrm{c}}$.
By integrating such a data-driven observer with the LQR controller, the system aims to adapt to driver-specific steering habits and unmodeled vehicle behaviors, thereby enhancing lateral control precision and personalization.

\section{Proposed NeuroDOB Design}
This section presents the overall design and training process of the proposed NeuroDOB.
First, the theoretical structure and formulation of the NeuroDOB are described,
followed by the training methodology used to learn the compensatory control behavior. The architecture of the NeuroDOB is illustrated in Fig.~\ref{fig:DNN architecture}.

\begin{figure}[!ht]
    \centering
    \includegraphics[width=0.49\textwidth]{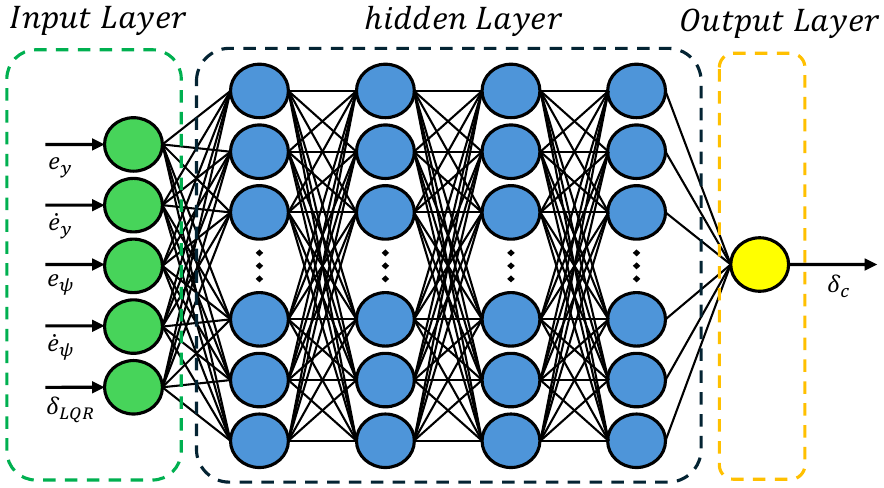}
    \caption{NeuroDOB architecture.}
    \label{fig:DNN architecture}
\end{figure}

\subsection{Design and Theoretical Framework}
The proposed NeuroDOB module enhances the baseline LQR controller by providing a real-time, driver-personalized steering compensation signal. 
Unlike conventional DOBs that rely on residuals between a nominal vehicle model and measured states, 
NeuroDOB leverages a deep neural network trained to imitate the compensation behavior of a surrogate driver. 
This enables the system to capture unmodeled dynamics and driver-specific control characteristics that classical DOBs cannot explicitly represent.

The NeuroDOB input vector at time step $k$, denoted $s[k] \in \mathbb{R}^n$, is defined as
\begin{equation}
    s[k] = \begin{bmatrix}
    e_y(k) & \dot{e}_y(k) & e_{\psi}(k) & \dot{e}_{\psi}(k) & \delta_{\mathrm{LQR}}(k)
    \end{bmatrix}^{\top}
\end{equation}
where $e_y$ and $\dot{e}_y$ are the lateral position error and its derivative, 
$e_{\psi}$ and $\dot{e}_{\psi}$ are the yaw error and its derivative, 
and $\delta_{\mathrm{LQR}}$ is the steering command generated by the baseline LQR controller.
Here, $\mathbb{R}^n$ represents an $n$-dimensional real vector space, where $n$ denotes the number of state variables or features considered in the model. 
Thus, $s[k]$ serves as the feature input vector to the neural network that estimates the compensatory steering command.

The network output is the compensation term $\delta_{\mathrm{c}}(k)$, which augments the baseline control command:
\begin{equation}
    \delta_{\mathrm{f}}(k) = \delta_{\mathrm{LQR}}(k) + \delta_{\mathrm{c}}(k)
\end{equation}
where $\delta_{\mathrm{f}}$ denotes the final steering command.

During training, the supervised target label is defined as the deviation between the surrogate driver’s actual steering and the baseline LQR command:
\begin{equation}
    \delta_{\mathrm{c}}(k) = \delta_{\mathrm{d}}(k) - \delta_{\mathrm{LQR}}(k)
\end{equation}
where $\delta_{\mathrm{d}}$ represents the driver's actual steering input.
The NeuroDOB is implemented as a fully connected deep neural network (DNN)\cite{ref22, ref23, ref24}.
In general, an $L$-layer feedforward DNN performs the recursive mapping:
\begin{equation}
\begin{aligned}
    h^{(0)} &= s[k] \in \mathbb{R}^{n}, \\
    z^{(\ell)} &= W^{(\ell)} h^{(\ell-1)} + b^{(\ell)}, \quad \ell = 1, \dots, L, \\
    \hat{z}^{(\ell)} &= \mathrm{BN}(z^{(\ell)}), \\
    h^{(\ell)} &= \mathrm{Dropout}\!\left(\tanh(\hat{z}^{(\ell)})\right), \quad \ell = 1, \dots, L-1, \\
    \bar{\delta}_c &= f_{\theta}(s[k]) = W^{(L)} h^{(L-1)} + b^{(L)}.
\end{aligned}
\label{eq:dnn_mapping}
\end{equation}
Here, $W^{(\ell)}$ and $b^{(\ell)}$ are the learnable weight matrices and bias vectors,
$\mathrm{BN}(\cdot)$ denotes batch normalization, and $\tanh$ is the nonlinear activation function.
The trainable parameter set is defined as $\theta = \{W^{(\ell)}, b^{(\ell)}\}_{\ell=1}^{L}$\cite{ref25}.
The implemented architecture consists of four hidden layers, each with 64 neurons,
with batch normalization and 0.2 dropout applied after each hidden layer.

\begin{theorem}
Consider the discrete-time vehicle lateral error dynamics:
\begin{equation}
    x[k+1] = \Phi x[k] + \Gamma \delta_{\mathrm{LQR}}[k] + \Gamma \delta_{\mathrm{c}}[k] + \Gamma_2 \dot{\psi}_{des}
\end{equation}
where $\delta_{\mathrm{LQR}}[k] = -K x[k]$ is the baseline LQR control law and $\delta_{\mathrm{c}}[k] = f_{\theta}(s[k])$ is the NeuroDOB compensation.  
If the learned compensation term $f_{\theta}(s[k])$ is bounded such that
\begin{equation}
    \| f_{\theta}(s[k]) \| \leq \epsilon_1,
\end{equation}
for some $\epsilon_1 > 0$, then the closed-loop system remains \textit{practically stable}, i.e.,
\begin{equation}
    \limsup_{k \to \infty} \| x[k] \| \leq \eta(\epsilon_1),
\end{equation}
where $\eta(\epsilon_1)$ is a class-$\mathcal{K}$ function dependent on $\epsilon_1$.  
\end{theorem}

\begin{proof}
Let the closed-loop nominal matrix under the baseline LQR be
\[
\Phi_{cl} \mathrel{\mathop:}= \Phi - \Gamma K,
\]
and assume \(\Phi_{cl}\) is Schur (all eigenvalues inside the unit disk), which holds because \(K\) is the LQR gain for the nominal model.

Since \(\Phi_{cl}\) is Schur, there exists a unique positive definite matrix \(P=P^\top\succ 0\) solving the discrete-time Lyapunov equation
\begin{equation}
    \Phi_{cl}^\top P \Phi_{cl} - P = -Q_0,
    \label{eq:lyap}
\end{equation}
for any chosen \(Q_0 = Q_0^\top \succ 0\). Define the quadratic Lyapunov function \(V(x)=x^\top P x\)\cite{ref26, ref27}.

The closed-loop dynamics including NeuroDOB compensation and disturbance are
\[
x[k+1] = \Phi_{cl} x[k] + \Xi[k],
\]
here the aggregated perturbation term is defined as
\[
\Xi[k] = \Gamma w[k] + \Gamma_2 \dot{\psi}_{des}[k].
\]
Where we denote the aggregated perturbation
\[
w[k] \mathrel{\mathop:}= f_{\theta}(s[k]).
\]
(We also keep the explicit disturbance \(\dot{\psi}_{des}\) term separate.) By assumption let \(w[k]\) and \(\dot{\psi}_{des}\) be bounded:
\[
\|w[k]\| \le \epsilon_1,\qquad \|\dot{\psi}_{des}[k]\| \le \epsilon_2,
\]
for some \(\epsilon_1, \epsilon_2>0\).

Compute the one-step Lyapunov increment:
\begin{equation}
\begin{split}
V(x[k+1]) - V(x[k])
&= x[k]^\top (\Phi_{cl}^\top P \Phi_{cl} - P)x[k]  \\
&\quad + 2\,x[k]^\top \Phi_{cl}^\top P\,\Xi[k]  \\
&\quad + \Xi[k]^\top P\,\Xi[k],
\end{split}
\label{eq:Vexpand_simple}
\end{equation}

Using~\eqref{eq:lyap}, this becomes
\begin{equation}
\begin{split}
V(x[k+1]) - V(x[k])
&= -\,x[k]^\top Q_0 x[k] \\
&\quad + 2\,x[k]^\top \Phi_{cl}^\top P\,\Xi[k] \\
&\quad + \Xi[k]^\top P\,\Xi[k].
\end{split}
\label{eq:Vincrement}
\end{equation}

We bound the cross term and the last quadratic term using norm inequalities. Let
\[
\begin{aligned}
\alpha &\mathrel{\mathop:}= \|\Phi_{cl}^\top P \Gamma\|, \quad
\beta  \mathrel{\mathop:}= \|\Phi_{cl}^\top P \Gamma_2\|, \\
\gamma &\mathrel{\mathop:}= \|P^{1/2} \Gamma\|, \quad
\delta \mathrel{\mathop:}= \|P^{1/2} \Gamma_2\|.
\end{aligned}
\]
Then from \eqref{eq:Vincrement} we have
\begin{align*}
V(x[k+1]) - V(x[k])
&\le - \lambda_{\min}(Q_0) \|x[k]\|^2 \\
&\quad + 2\|x[k]\|(\alpha \|w[k]\| + \beta \|\dot{\psi}_{des}[k]\|) \\
&\quad + (\gamma \|w[k]\| + \delta \|\dot{\psi}_{des}[k]\|)^2.
\end{align*}

Apply the scalar inequality \(2ab \le \mu a^2 + \mu^{-1} b^2\) (valid for any \(\mu>0\)) to the cross term with \(a=\|x[k]\|\) and \(b=\alpha\|w[k]\|+\beta\|\dot{\psi}_{des}[k]\|\). Choose \(\mu = \lambda_{\min}(Q_0)/2 >0\). Then there exists constants \(c_1,c_2>0\) (depending only on \(\Phi_{cl},\Gamma,\Gamma_2,P,Q_0\)) such that
\[
\begin{aligned}
V(x[k+1]) - V(x[k])
&\le -\tfrac{1}{2}\lambda_{\min}(Q_0)\|x[k]\|^2 \\
&\quad + c_1\|w[k]\|^2 + c_2\|\dot{\psi}_{des}[k]\|^2.
\end{aligned}
\]

Using the uniform bounds \(\|w[k]\|\le\epsilon_1\) and \(\|\dot{\psi}_{des}[k]\|\le\epsilon_2\), we obtain
\[
V(x[k+1]) - V(x[k]) \le -\tfrac{1}{2}\lambda_{\min}(Q_0)\|x[k]\|^2 + c,
\]
where \(c \mathrel{\mathop:}= c_1\epsilon_1^2 + c_2\epsilon_2^2\). Rearranging gives
\[
\tfrac{1}{2}\lambda_{\min}(Q_0)\|x[k]\|^2 \le V(x[k]) - V(x[k+1]) + c.
\]

Sum both sides from \(k=0\) to \(N-1\) and telescope the left-hand side to obtain a bound on the time-average of \(\|x[k]\|^2\). Standard arguments for difference inequalities (or invoking invariance arguments for practical stability) yield that the state norm is ultimately bounded:
\[
\limsup_{k\to\infty}\|x[k]\| \le \eta(\epsilon_1,\epsilon_2),
\]
for some continuous function \(\eta(\cdot,\cdot)\) with \(\eta(0,0)=0\). In particular, for fixed \(\epsilon_2\) the bound \(\eta\) is a class-\(\mathcal{K}\) function of \(\epsilon_1\), which proves practical stability as stated.
\end{proof}

\begin{corollary}
If the learned compensation perfectly approximates the unknown disturbance, i.e., $f_{\theta}(s[k]) = -\Gamma^{-1}\Gamma_2 \dot{\psi}_{des}[k]$,  
then the closed-loop error dynamics reduce to
\begin{equation}
    x[k+1] = (\Phi - \Gamma K) x[k],
\end{equation}
and the system achieves asymptotic stability as guaranteed by the baseline LQR.
\end{corollary}

\begin{remark}
The above theorem indicates that the proposed NeuroDOB does not compromise the stability ensured by the LQR controller.  
Instead, it adaptively improves transient response and steady-state performance by compensating for unmodeled disturbances and driver-specific behaviors.  
Hence, the LQR–NeuroDOB structure can be interpreted as a stable yet adaptive control framework combining model-based and data-driven reasoning.
\end{remark}

\subsection{Training Objective and Optimization Strategy}
The learning objective of the NeuroDOB is to minimize the mean squared error (MSE) between the predicted and target compensation values:
\begin{equation}
    \mathcal{L}(\theta) 
    = \frac{1}{N} \sum_{k=1}^{N} 
    \left( \bar{\delta}_{\mathrm{c}}(k) - \delta_{\mathrm{c}}(k) \right)^2
\end{equation}
where $\theta$ represents the neural network parameters, 
and $N$ is the total number of training samples\cite{ref28}.

The model is trained using the Adam optimizer with a learning rate of $10^{-3}$ 
and $\ell_2$ regularization (weight decay) coefficient $\lambda$. 
A learning rate scheduler (ReduceLROnPlateau) is employed to enhance convergence, 
reducing the learning rate by a factor of 0.5 if the validation loss does not improve for 10 epochs.
Early stopping is also applied to terminate training when the validation loss fails to improve by more than $\Delta_{\min}=10^{-5}$ over 50 consecutive epochs.
Additionally, all input and output variables are standardized using the dataset statistics to ensure consistent scaling across both training and inference phases.

\subsection{Comparison between Conventional DOB and NeuroDOB}

The conventional disturbance observer (DOB) is a model-based estimator that reconstructs the disturbance \(d(k)\) by comparing the measured output with the nominal model response\cite{ref29, ref30}.  
The discrete-time lateral dynamics model of the vehicle can be expressed as:

\begin{equation}
    x[k+1] = \Phi x[k] + \Gamma u[k] + \Gamma_2 \dot{\psi}_{des}[k] + B_{d} d[k]
\end{equation}

The DOB estimates the disturbance by filtering the residual between the measured output and the nominal model output through a low-pass filter \(Q\):

\begin{equation}
    \hat{d}(k) = Q\big(y(k) - C_n(\Phi_n x(k-1) + \Gamma_n u(k-1))\big)
\end{equation}

The compensated control input is defined as:

\begin{equation}
    u_f(k) = u(k) - \hat{d}(k)
\end{equation}

Here, \(Q\) is a low-pass filter that determines the disturbance estimation bandwidth.  
The performance of this conventional DOB strongly depends on the accuracy of the nominal model and the design quality of the filter.  
Therefore, when model mismatch or nonlinear driver-specific behaviors exist, the disturbance compensation performance may deteriorate. These limitations motivate the development of a data-driven observer that can learn nonlinear and driver-specific disturbance patterns directly from driving data.

In contrast, the proposed NeuroDOB replaces explicit model-based residual reconstruction with a data-driven approach that directly learns the mapping from observed states and baseline controller outputs to the required compensation term.  
Formally, at each time step \(k\), the NeuroDOB defines the compensation as:

\begin{equation}
    \delta_{\mathrm{c}}(k) = f_{\theta}\big(s[k]\big)
\end{equation}

where \(s[k]\in\mathbb{R}^n\) denotes the state feature vector containing tracking errors and the baseline LQR command,  
and \(f_{\theta}(\cdot)\) is a deep neural network (DNN) parameterized by \(\theta\).  
The final command is then obtained as:

\begin{equation}
    \delta_{\mathrm{f}}(k) = \delta_{\mathrm{LQR}}(k) + \delta_{\mathrm{c}}(k)
\end{equation}

Unlike the analytical DOB structure that computes residuals through a fixed filter \(Q(z)\),  
the NeuroDOB exploits the universal function approximation capability of neural networks~\cite{ref25} to represent the complex nonlinear mapping:

\begin{equation}
    f_{\theta}: \mathbb{R}^n \;\longrightarrow\; \mathbb{R}
\end{equation}

which implicitly embeds disturbance dynamics, driver-specific behaviors, and unmodeled vehicle dynamics into the learned function space.  
The overall architecture of the proposed NeuroDOB is illustrated in Fig.~\ref{fig:NeuroDOB architecture}.

\begin{figure}[!ht]
    \centering
    \includegraphics[width=0.49\textwidth]{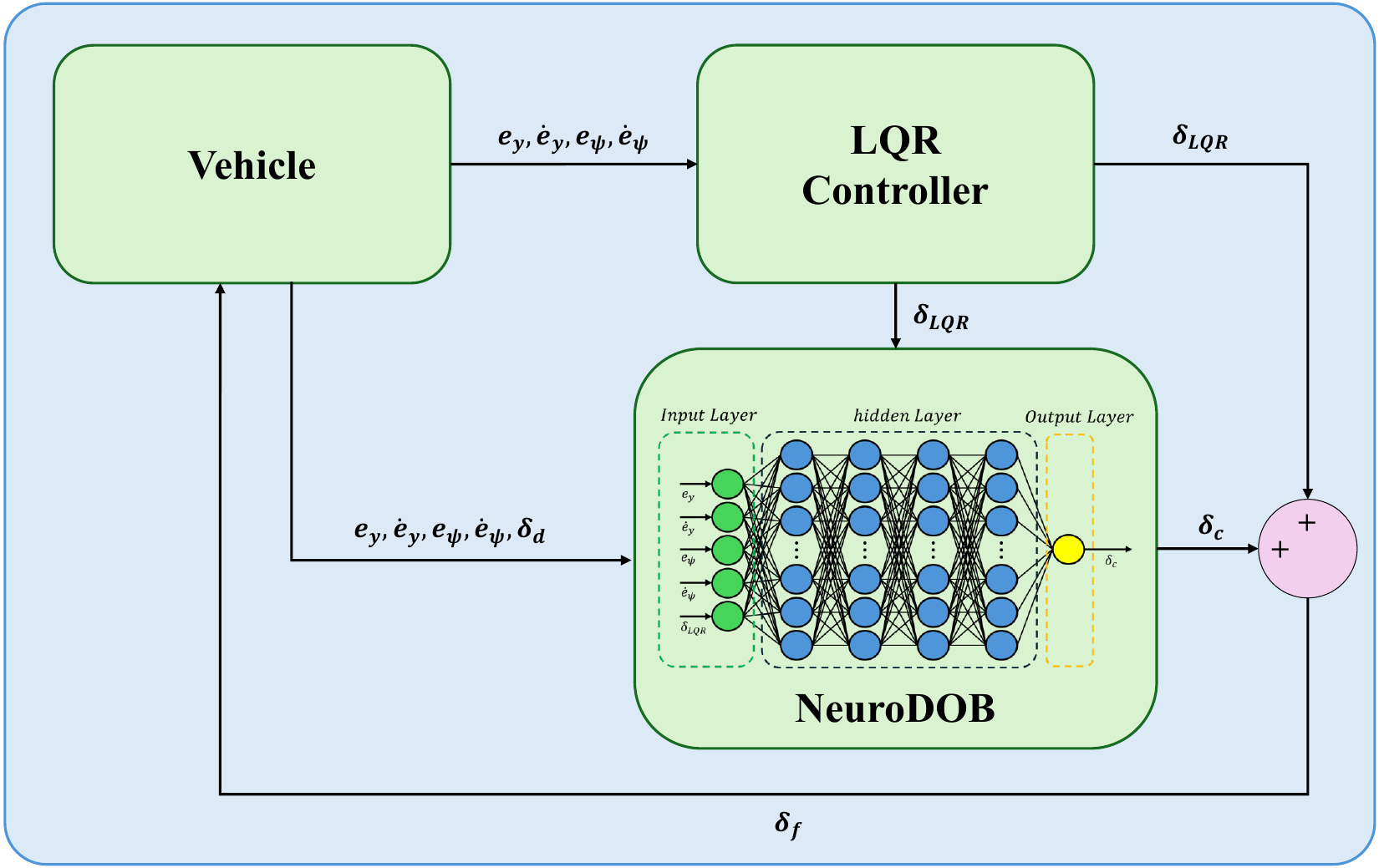}
    \caption{Architecture of the proposed NeuroDOB showing the data-driven compensation pathway integrated with the baseline LQR controller.}
    \label{fig:NeuroDOB architecture}
\end{figure}

Therefore, while the DOB focuses on disturbance estimation through model residuals and linear filtering,  
the NeuroDOB focuses on disturbance compensation learning directly from data by modeling the nonlinear mapping between error states and steering compensation.  
This data-driven approach enables more flexible and personalized control performance.

\section{Simulation setup and Result}
This section presents the simulation settings and results for three distinct test cases, each corresponding to a specific road geometry. By systematically varying the road shapes in simulation, we examine not only performance outcomes but also the underlying impact of training data diversity and road geometry on the neural network’s learning and generalization capabilities.

\subsection{Simulation setting}
All simulations are conducted using a high-fidelity CarSim simulation environment configured to model realistic vehicle dynamics and sensor feedback. The test vehicle parameters and control sampling intervals are set according to the values detailed in Section~2 and Table~\ref{tab:params}. The $V_x$ is held constant at 50\,km/h throughout all simulations, reflecting a fixed-speed scenario for robust evaluation. Three distinct road maps are utilized to evaluate the proposed controller. Road Map 1 is used for data collection and training, Road Map 2 is employed for cross-validation and generalization testing, and Road Map 3 is used for training with slight geometry and curvature differences compared to Road Map 2. Fig.~\ref{fig:road123} illustrate the layouts of these three maps. And the three experiment cases are summarized in Table~\ref{tab:case}.

\begin{table}[!ht]
\caption{Three Simulation Cases}
\centering
\captionsetup{justification=centering}
\small
\resizebox{\columnwidth}{!}{%
\begin{tabular}{|l|c|c|c|}
\hline
 & \textbf{Case 1} & \textbf{Case 2} & \textbf{Case 3} \\
\hline
Training Road Map & Road Map 1 & Road Map 1 & Road Map 3 \\
\hline
Validation Road Map & Road Map 1 & Road Map 2 & Road Map 2 \\
\hline
\end{tabular}
}
\label{tab:case}
\end{table}

\begin{table}[!ht]
    \caption{Test Vehicle Parameters}
    \centering
    \begin{tabular}{ll}
        \hline
        \textbf{Parameter}          & \textbf{Value}                \\
        \hline
        Vehicle mass $(m)$          & 1,274\,kg                    \\
        Yaw inertia $(I_{z})$      & 1,523\,kg·m$^2$              \\
        Front/rear tire distance $(l_f/l_r)$ & 1.016 /\, 1.562\,m   \\
        Cornering stiffness $(C_{af}/C_{ar})$ & 118{,}800 / 165{,}300\,N/rad \\
        Control sampling time $(T_c)$& 0.01\,s                      \\
        \hline
    \end{tabular}
    \label{tab:params}
\end{table}

\begin{figure*}[t]
    \centering
    \includegraphics[width=\textwidth]{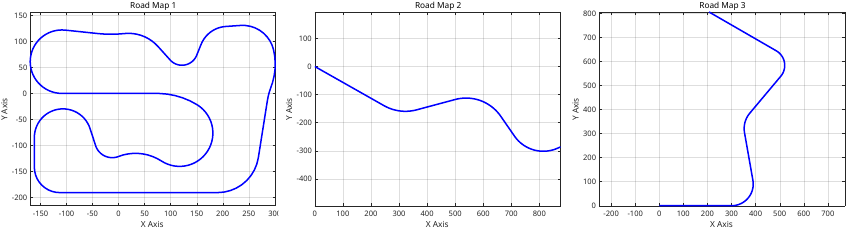}
    \caption{Road map shapes used in this study. Road Map 1 was used for training data collection and controller evaluation, Road Map 2 for independent testing, and Road Map 3 with modified geometry for additional training.}
    \label{fig:road123}
\end{figure*}

Training data is acquired by operating the CarSim simulated vehicle using the embedded driver controller on Road Map 1 and Road Map 3 for 100 seconds. During this period, four error states ($e_y$, $\dot{e}_y$, $e_{\psi}$, $\dot{e}_{\psi}$), $\delta_{\mathrm{LQR}}$, and  $\delta_{\mathrm{d}}$ are synchronously logged at each control step. The collected samples are then unified into a single dataset. Prior to training, standard preprocessing procedures are applied including synchronization, outlier removal, and normalization.

\subsection{Simulation results}
In this section, we present the key results obtained from the experimental evaluations conducted on different road maps. The performance of the baseline LQR controller, the NeuroDOB-compensated controller, and the CarSim embedded driver are compared in terms of lateral and yaw tracking errors, as well as control input behavior. These results demonstrate the effectiveness of the proposed NeuroDOB approach in enhancing lateral control accuracy and adaptability to varying road conditions. Detailed analysis follows, supported by time series plots and quantitative metrics to highlight the strengths and limitations observed.

\begin{itemize}
\item \textbf{Case 1}: The NeuroDOB is trained and validated on Road Map 1 to evaluate its fundamental learning capability and trajectory tracking performance.

\item \textbf{Case 2}: The NeuroDOB trained on Road Map 2 is validated on Road Map 1, which has different curvature and geometry, to analyze its generalization capability against varying road conditions.

\item \textbf{Case 3}: The NeuroDOB trained on Road Map 3, which has similar curvature and geometry to Road Map 2, is validated on Road Map 2 to assess the importance of road geometry characteristics and diversity of training data.
\end{itemize}

\begin{figure}[!ht]
    \centering
    \includegraphics[width=0.45\textwidth]{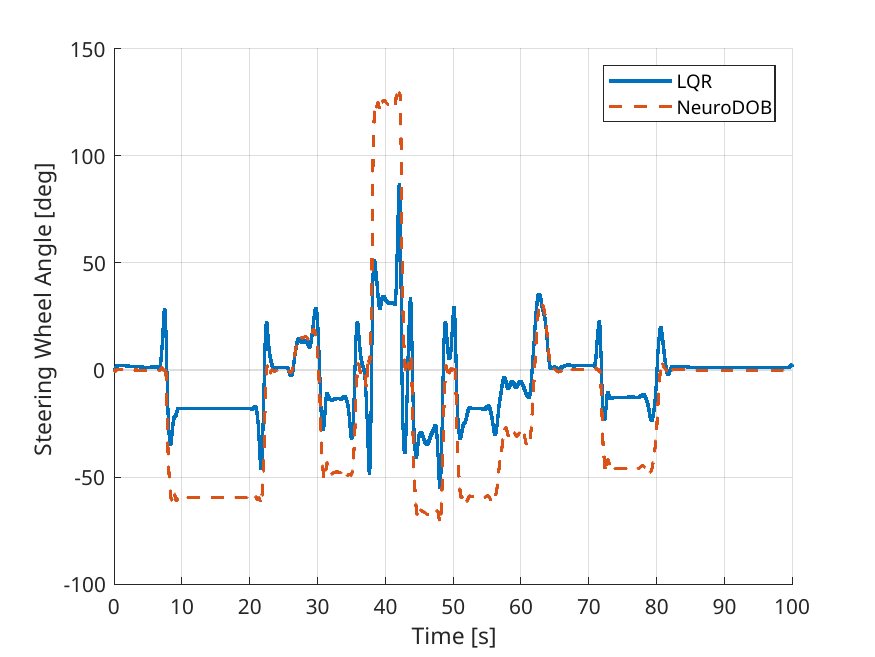}
    \caption{Comparison of steering inputs generated by the LQR controller and the NeuroDOB-compensated controller on Road Map 1.}
    \label{fig:lqr_vs_neurodob_control_input_map1}
\end{figure}

\subsubsection{Simulation Case 1}
Fig~\ref{fig:lqr_vs_neurodob_control_input_map1} compares the time series of steering inputs generated by the baseline LQR controller and by the LQR controller augmented with NeuroDOB compensation on Road Map 1. The LQR controller can exhibit delayed response and overshoot, particularly in challenging sections. In contrast, the steering input compensated by NeuroDOB adjusts more sharply and proactively in sharp curves and dynamic transition segments, thereby more closely emulating the subtle steering style of a CarSim embedded driver. Such additional compensatory commands enable the controller to respond more quickly and robustly in demanding road segments. The overlay of both signals clearly demonstrates the improvement in dynamic adaptability and agility afforded by data-driven compensation.

Fig~\ref{fig:roadmap1c0c1} presents a comprehensive comparison of the time series of lateral position error($e_y$) and heading error($e_\psi$) for three controllers CarSim embedded driver, pure LQR, and LQR with NeuroDOB compensation during operation on Road Map 1. 

\begin{figure}[!ht]
    \centering
    \includegraphics[width=0.45\textwidth]{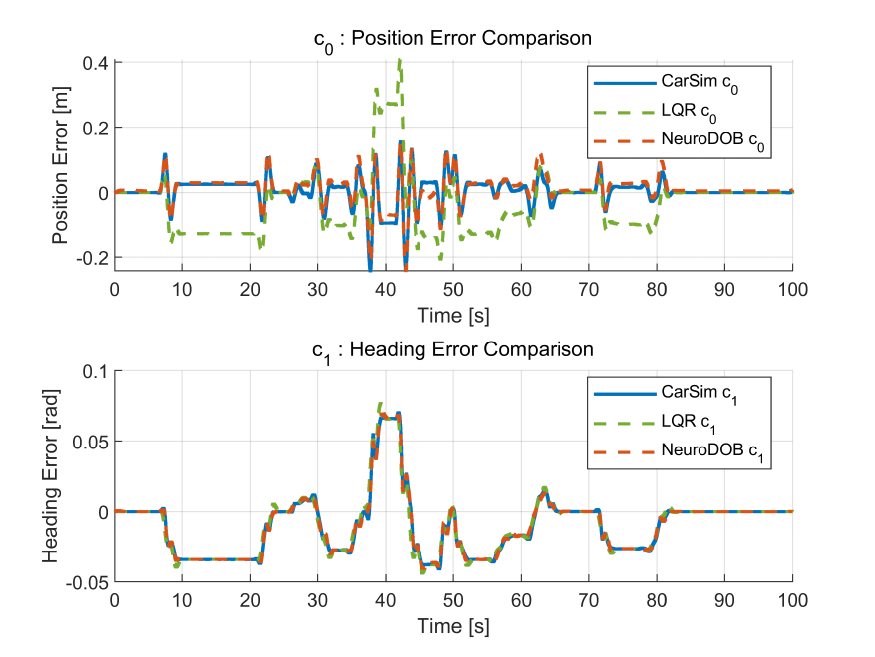}
    \caption{Comparison of $e_y$ and $e_\psi$ for CarSim embedded driver, LQR, and NeuroDOB controllers on Road Map 1.}
    \label{fig:roadmap1c0c1}
\end{figure}

The CarSim embedded driver demonstrates the lowest error magnitudes and exceptionally smooth transitions in both $e_y$ and $e_\psi$, even through curves and dynamic path changes, establishing an effective “optimal reference” for controller evaluation. The pure LQR controller maintains overall lane tracking but exhibits noticeably increased error peaks and oscillatory behaviors, particularly during sharp curves and complex transitions, highlighting its limited adaptability to unmodeled vehicle dynamics and individual driver characteristics. 

In comparison, the NeuroDOB-compensated controller significantly reduces both the magnitude and oscillations of the errors relative to pure LQR, with performance approaching that of the CarSim reference, especially in challenging segments. The NeuroDOB controller also produces steering patterns more similar to those of the CarSim embedded driver, underscoring its data-driven adaptability and effectiveness in real-time trajectory tracking.

This unified plot enables direct, quantitative assessment of error reduction and performance improvements across all methods, using the CarSim reference as a standard for evaluating both adaptive similarity and overall control enhancement.

\begin{table}[!ht]
\caption{Simulation Case 1 RMSE Comparison}
\centering
\captionsetup{justification=centering}
\small
\begin{tabular}{|l|c|c|c|}
\hline
 & \textbf{LQR} & \textbf{NeuroDOB} & \textbf{Change} \\
\hline
$e_y$ [m] & 0.1096 & 0.0150 & 86.31\% $\Downarrow$ \\
\hline
$e_\psi$ [rad] & 0.0232 & 0.0232 & 0.0\%  \\
\hline
\end{tabular}
\label{tab:road map 1 rmse_comparison}
\end{table}

Table~\ref{tab:road map 1 rmse_comparison} presents the root mean square error (RMSE) comparison between the conventional LQR controller and the proposed NeuroDOB controller for two key lateral control metrics: $e_y$ and $e_\psi$. The introduction of NeuroDOB compensation results in a significant reduction of approximately 86.3\% in $e_y$ RMSE, indicating substantial improvement in lateral path tracking accuracy. In contrast, $e_\psi$ RMSE exhibits a minimal increase of about 0.03\%, suggesting that the proposed method primarily enhances lateral control performance without compromising yaw stability. These quantitative results clearly demonstrate that the NeuroDOB compensation architecture significantly improves lateral control precision compared to the baseline LQR controller while maintaining comparable yaw performance.

From a dual-system perspective, the LQR (System 1) maintains nominal control, while the NeuroDOB (System 2) monitors the residual error and injects reflective compensation.
The observed reduction of RMSE by over 80\% indicates that the interaction between intuitive (model-based) and reflective (data-driven) reasoning yields synergistic control performance. \\

\subsubsection{Simulation Case 2}

Fig~\ref{fig:lqr_vs_neurodob_control_input_map2} shows the time series of control inputs from the baseline LQR controller and the LQR controller augmented with NeuroDOB compensation on Road Map 2, which was not used during training. The LQR input remains relatively stable on straight sections but exhibits overshoot at the beginning of curves. In contrast, the NeuroDOB-compensated control input shows slight overshoot as well, but performs comparatively dynamic and fine-grained adjustments in sharp curves and path transition regions, effectively reflecting the road environment. This demonstrates NeuroDOB’s adaptive capability to generalize the learned compensation signals to novel road scenarios where pure model-based controllers show limitations, thereby improving lateral vehicle control. The figure visually highlights how data-driven compensation contributes to enhanced controller adaptability and reduced tracking error under more challenging driving conditions.

\begin{figure}[!ht]
    \centering
    \includegraphics[width=0.45\textwidth]{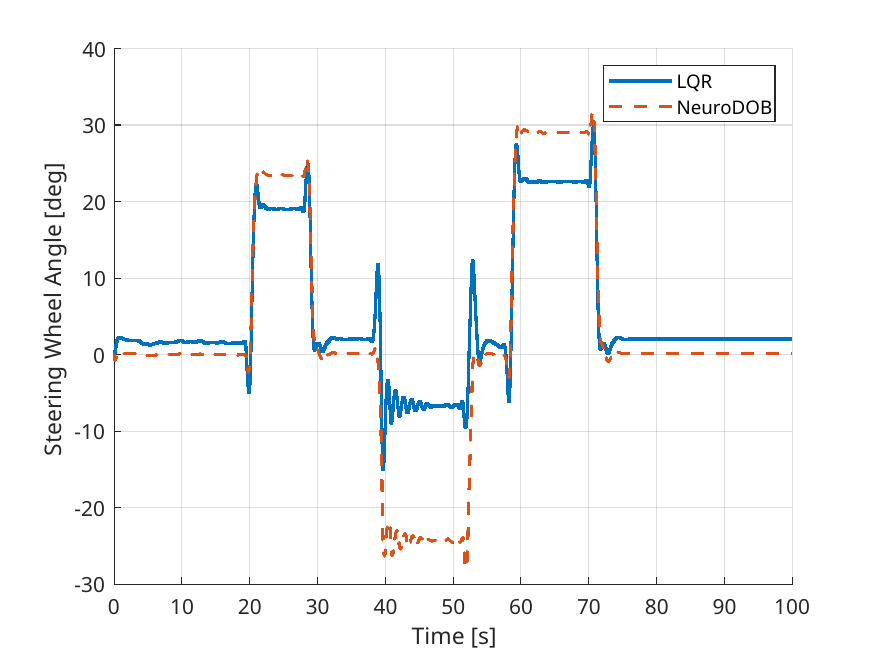}
    \caption{Steering input comparison between the LQR controller and NeuroDOB-compensated controller on Road Map 2.}
    \label{fig:lqr_vs_neurodob_control_input_map2}
\end{figure}

Fig~\ref{fig:roadmap2c0c1} presents a unified comparison of the time series of $e_y$ and $e_\psi$ for the CarSim embedded driver, the pure LQR controller, and the LQR controller with NeuroDOB compensation on Road Map~2, which was not included in the training data.

\begin{figure}[!ht]
    \centering
    \includegraphics[width=0.45\textwidth]{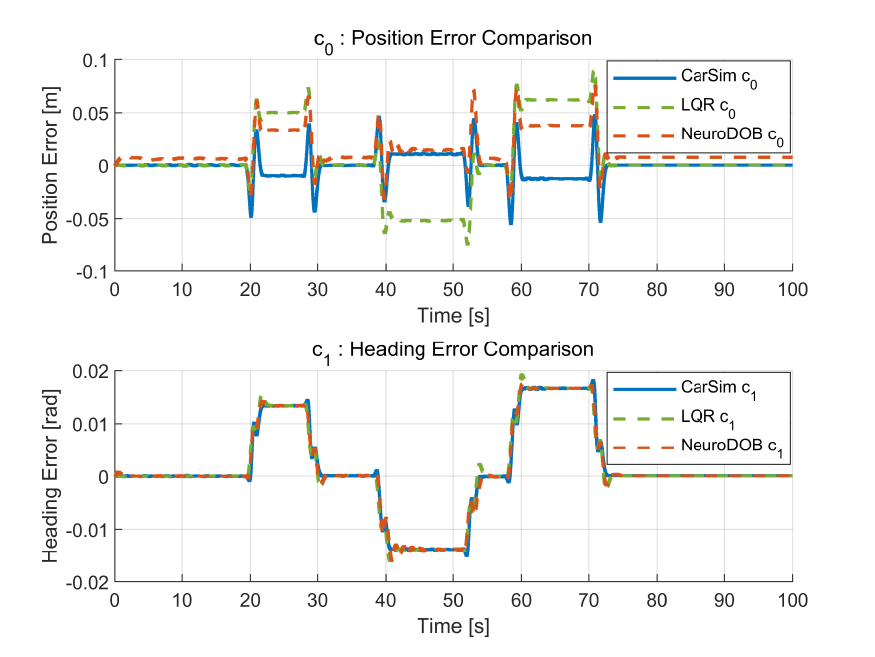}
    \caption{Comparison of $e_y$ and $e_\psi$ for CarSim embedded driver, LQR, and NeuroDOB controllers on Road Map 2.}
    \label{fig:roadmap2c0c1}
\end{figure}

The CarSim embedded driver acts as a surrogate skilled human driver, maintaining generally low error levels but exhibiting sharp steering adjustments, particularly at curve entry points, in response to the significantly increased curvature variation and environmental complexity of the test map. The pure LQR controller maintains relatively stable control performance on straight road segments. However, during sharp curves and path transitions, both the error magnitude and oscillations increase, revealing its limitations in adapting to unmodeled vehicle dynamics and complex road geometries.

The LQR controller with NeuroDOB compensation achieves tracking performance closer to the CarSim reference trajectory, reducing both the magnitude and oscillations of the errors compared to the pure LQR controller and enhancing adaptability to novel road conditions. Nevertheless, some discrepancies from the CarSim trajectory remain, indicating the structural limitation of relying on a single trained neural network. 

\begin{table}[!ht]
\caption{Simulation Case 2 RMSE Comparison}
\centering
\captionsetup{justification=centering}
\small
\begin{tabular}{|l|c|c|c|}
\hline
 & \textbf{LQR} & \textbf{NeuroDOB} & \textbf{Change} \\
\hline
$e_y$ [m] & 0.0374 & 0.0225 & 39.93\% $\Downarrow$ \\
\hline
$e_\psi$ [rad] & 0.0083 & 0.0083 & 0.0\%  \\
\hline
\end{tabular}
\label{tab:road map 2 rmse_comparison}
\end{table}

Table~\ref{tab:road map 2 rmse_comparison} presents RMSE comparison between the conventional LQR controller and the proposed NeuroDOB controller for Road Map~2, which was not included in the training data. For $e_y$, NeuroDOB achieves a reduction of approximately 39.9\% in RMSE compared to the pure LQR controller, indicating a moderate improvement in tracking accuracy under novel road geometry. In contrast, $e_\psi$ shows no change in RMSE, suggesting that the proposed method maintains yaw stability performance equivalent to that of the baseline LQR on this test map.

These results demonstrate that while NeuroDOB still provides tangible improvement in lateral path tracking for previously unseen road conditions, the gain is smaller than that observed on Road Map~1. This performance gap is primarily attributed to the curvature distribution differences between training and test maps, as discussed in the following section. Nevertheless, the findings confirm that the NeuroDOB compensation mechanism retains its ability to enhance lateral control accuracy without degrading yaw stability, even in more complex and unfamiliar driving scenarios.

Fig.~\ref{fig:curvature_histogram} shows the road curvature histograms for Road Map 1, which was used for training the NeuroDOB, and Road Map 2, to which the trained NeuroDOB was directly applied without additional retraining. It can be clearly observed that there is a significant difference in curvature distributions between the two maps. Road Map 1 features a wide and complex distribution with frequent and sharp curvature changes, whereas Road Map 2 exhibits a relatively limited curvature range concentrated within specific intervals.

\begin{figure}[!ht]
    \centering
    \includegraphics[width=0.5\textwidth]{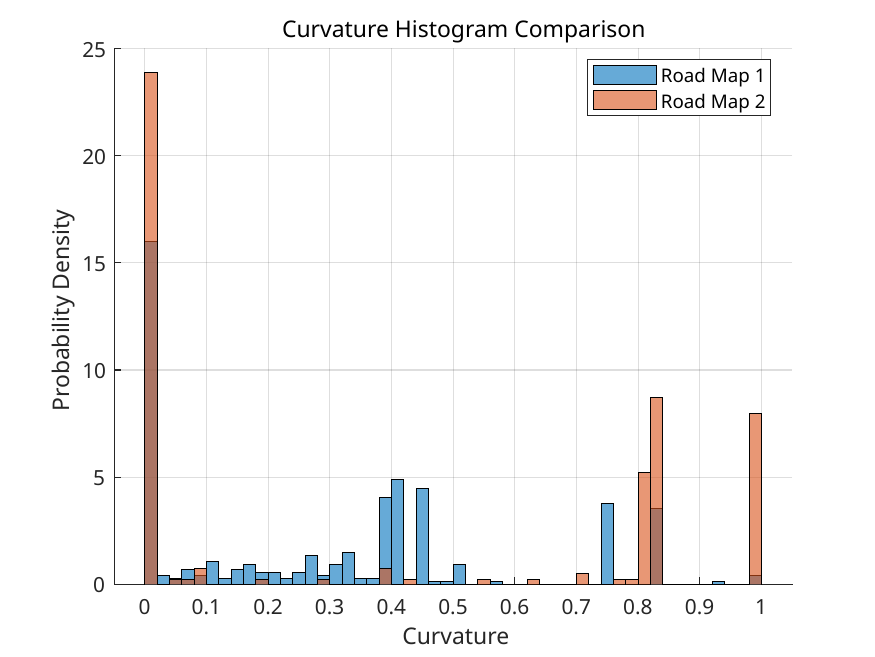}
    \caption{Comparison of curvature distributions between Road Map 1 and Road Map 2.}
    \label{fig:curvature_histogram}
\end{figure}

This difference in curvature distributions explains why the NeuroDOB trained only on Road Map 1 cannot fully replicate the control behavior of the CarSim embedded controller on Road Map 2. In other words, the lack of diversity in the training data limits the model’s generalization capability and causes performance degradation under curvature conditions that were not encountered during training.

Accordingly, we present additional experimental results where the NeuroDOB is trained on a track whose road geometry and curvature characteristics are similar, though not identical in distribution and length, to Road Map 2, and then the trained model is applied to driving on Road Map 2.
\\
\subsubsection{Simulation Case 3}
To evaluate the effect of road curvature and geometric similarity on the performance of NeuroDOB, a new track designated as Road Map 3 was constructed for experimental validation. While Road Map 1 and Road Map 2 exhibit distinct differences in curvature distributions and road geometries, Road Map 3 shows only slight variations in road shape and curvature characteristics compared to Road Map 2. In this experiment, NeuroDOB was trained using data collected from Road Map 3, and the trained model was then applied to driving on Road Map 2 to evaluate its performance. This approach aims to analyze the influence of training data diversity and curvature distribution on NeuroDOB’s learning efficiency and generalization capability.

Fig.~\ref{fig:lqr_vs_neurodob_control_input_map3} shows the steering input comparison between the baseline LQR controller and the NeuroDOB-compensated controller on Road Map 3. The results demonstrate that the NeuroDOB-augmented controller produces steering input patterns more closely aligned with those of the CarSim embedded driver controller, indicating enhanced adaptation to driver behavior and road conditions.

\begin{figure}[!ht]
    \centering
    \includegraphics[width=0.45\textwidth]{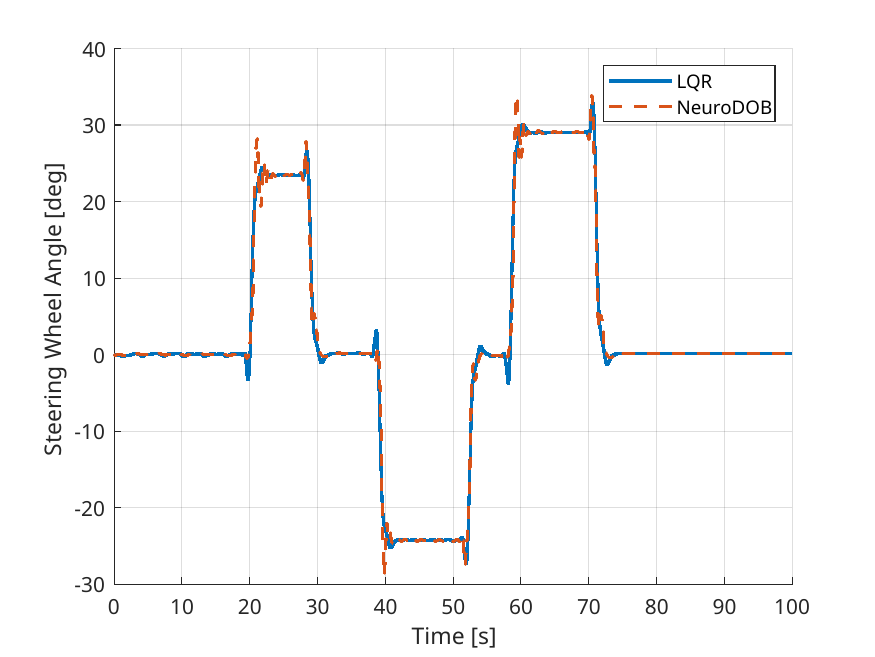}
    \caption{Steering input comparison between the LQR controller and NeuroDOB-compensated controller on Road Map 3.}
    \label{fig:lqr_vs_neurodob_control_input_map3}
\end{figure}

Fig.~\ref{fig:roadmap3c0c1} compares the time series of $e_y$ and $e_\psi$ for three controllers—the CarSim embedded driver, pure LQR, and LQR with NeuroDOB compensation—on Road Map 3. 

\begin{figure}[!ht]
    \centering
    \includegraphics[width=0.45\textwidth]{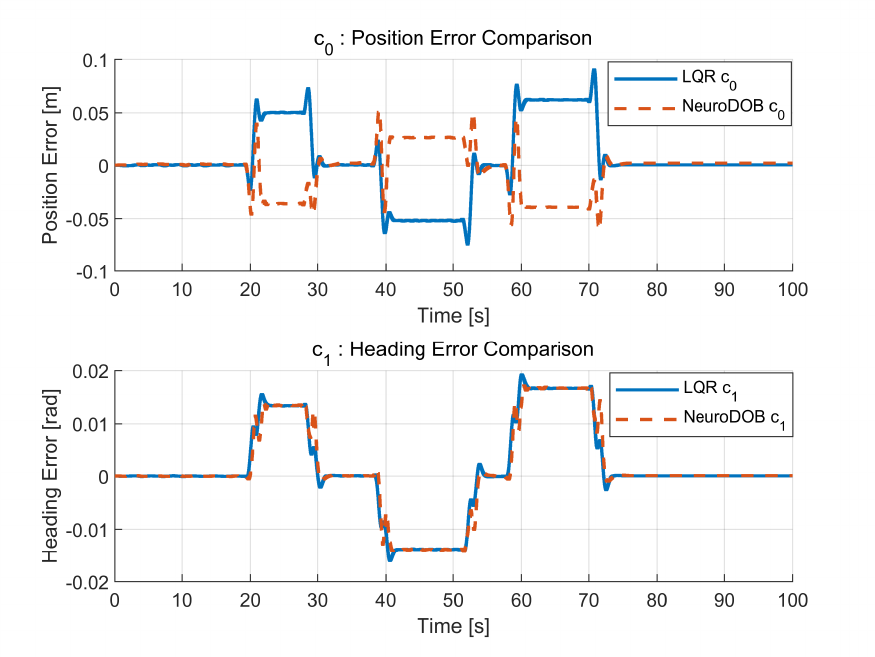}
    \caption{Comparison of $e_y$ and $e_\psi$ for CarSim embedded driver, LQR, and NeuroDOB controllers on Road Map 3.}
    \label{fig:roadmap3c0c1}
\end{figure}

Although the road geometry differs from the training conditions and thus the NeuroDOB-compensated controller does not achieve perfect tracking, it shows a clear improvement over the pure LQR controller. This result indicates a tendency of the NeuroDOB controller to mimic the driving behavior of the CarSim embedded driver. It also underscores the importance of training data diversity and curvature distribution in the learning process, demonstrating that NeuroDOB can generalize and adapt even under slight variations in road conditions. Overall, this analysis supports the effectiveness of NeuroDOB in enhancing lateral control performance beyond the baseline LQR.

\begin{table}[!ht]
\caption{Whole Simulation Cases RMSE Comparison}
\centering
\captionsetup{justification=centering}
\small
\resizebox{\columnwidth}{!}{%
\begin{tabular}{|c|c|c|c|c|}
\hline
\textbf{Simulation} & \textbf{State} & \textbf{LQR} & \textbf{NeuroDOB} & \textbf{Change} \\
\hline
\multirow{2}{*}{Case 1} 
    & $e_y$ [m] & 0.1096 & 0.0150 & 86.31\% $\Downarrow$ \\ \cline{2-5}
    & $e_\psi$ [rad] & 0.0232 & 0.0232 & 0.0\% \\
\hline
\multirow{2}{*}{Case 2}
    & $e_y$ [m] & 0.0374 & 0.0225 & 39.93\% $\Downarrow$ \\ \cline{2-5}
    & $e_\psi$ [rad] & 0.0083 & 0.0083 & 0.0\%  \\
\hline
\multirow{2}{*}{Case 3}
    & $e_y$ [m] & 0.0374 & 0.0173 & 53.64\% $\Downarrow$ \\ \cline{2-5}
    & $e_\psi$ [rad] & 0.0083 & 0.0083 & 0.0\%  \\
\hline
\end{tabular}
}
\label{tab:Three_Road_Map_rmse_comparison}
\end{table}

Table~\ref{tab:Three_Road_Map_rmse_comparison} summarizes the RMSE comparison between the conventional LQR controller and the proposed NeuroDOB controller across three different road maps. The first experiment involved training and testing NeuroDOB on Road Map 1, where it achieved approximately 86.31\% improvement in $e_y$. The second experiment evaluated the trained NeuroDOB from the first experiment on a different road map, Road Map 2, resulting in about 39.93\% improvement in $e_y$. However, this improvement was relatively smaller compared to that of the first experiment. This difference is attributed to variations in the distribution of training data and curvature of the evaluation roads, leading to a third experiment conducted for further analysis.

In the case 3, NeuroDOB was trained on Road Map 3, which features geometric and curvature differences from Road Map 2, and tested on Road Map 2. This resulted in approximately 53.64\% improvement in $e_y$, while $e_\psi$ remained stable without significant fluctuations. These results indicate that NeuroDOB effectively enhances lateral path tracking even in environments with slight deviations from training data while maintaining $e_\psi$ stability. Furthermore, compared to the previously reported results on Road Map 2, the improvement in $e_y$ on Road Map 3 was more pronounced, emphasizing the critical role of training data distribution and road curvature representation in achieving robust generalization.

\section{Experiment setup and Result}
To evaluate and verify the applicability of NeuroDOB in real vehicles and its effective operation under actual driving conditions, experiments were conducted utilizing driving data collected from an actual vehicle. This approach allowed us to confirm that the theoretical performance achieved in simulation environments is sustained during real-world road driving, and to empirically demonstrate that NeuroDOB accurately reflects driver-specific characteristics and dynamic changes in vehicle behavior across various driving scenarios, thereby successfully achieving personalized lateral control. For validation, a comparative analysis was performed among the pure LQR controller, the LQR controller with NeuroDOB integration, and the actual driver’s steering input.

\subsection{Experiment setup}
To conduct actual vehicle experiment, a Kia Niro was used as the test platform, as shown in Fig.~\ref{fig:niro}. 

\begin{figure}[!ht]
    \centering
    \includegraphics[width=0.45\textwidth]{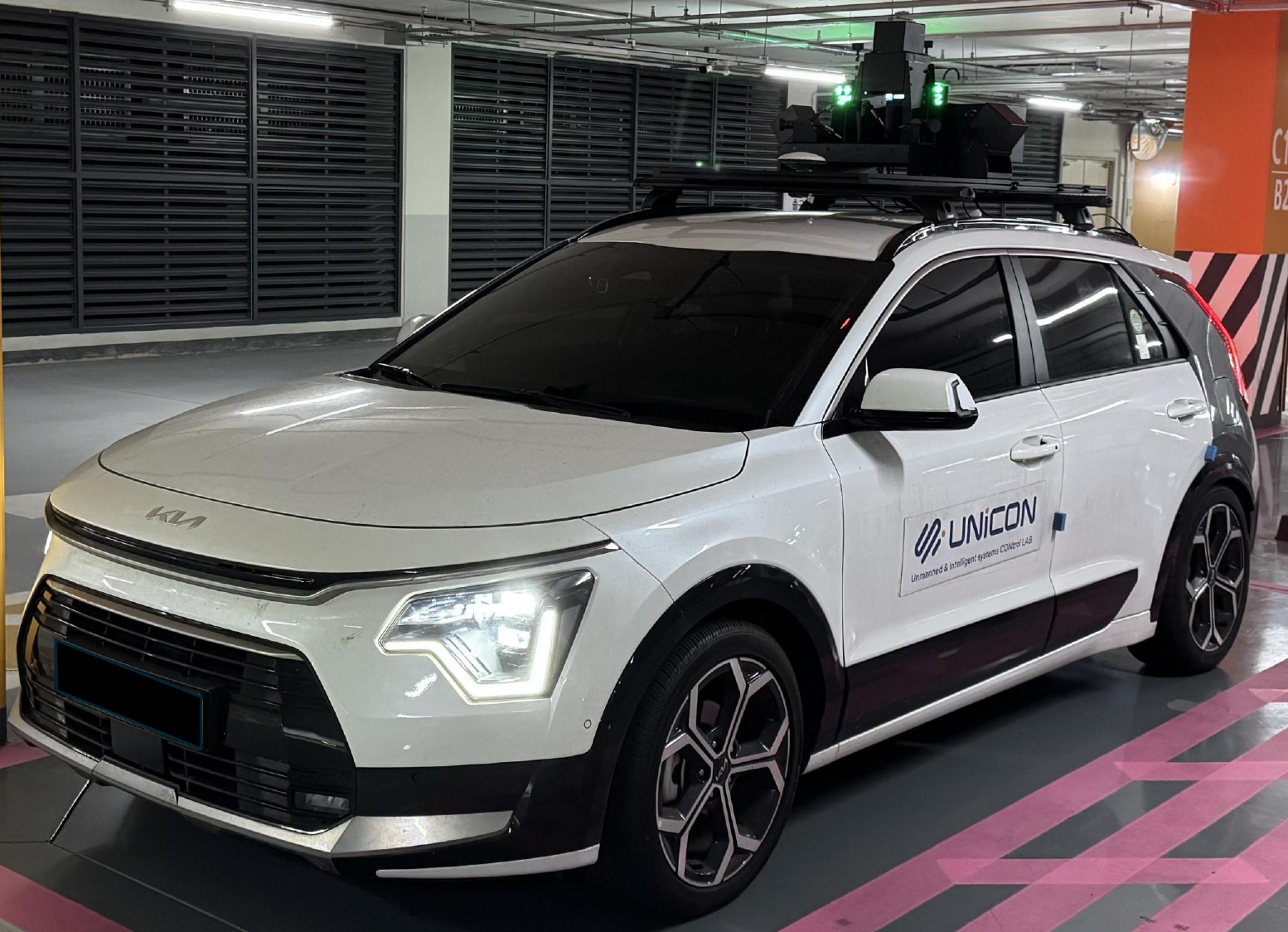}
    \caption{Vehicle used in the experiment.}
    \label{fig:niro}
\end{figure}

The tests were carried out on a highway section at speeds between 80 and 90 km/h, the actual road where the driving test was performed is shown in Fig.~\ref{fig:road}. 

\begin{figure}[!ht]
    \centering
    \includegraphics[width=0.45\textwidth]{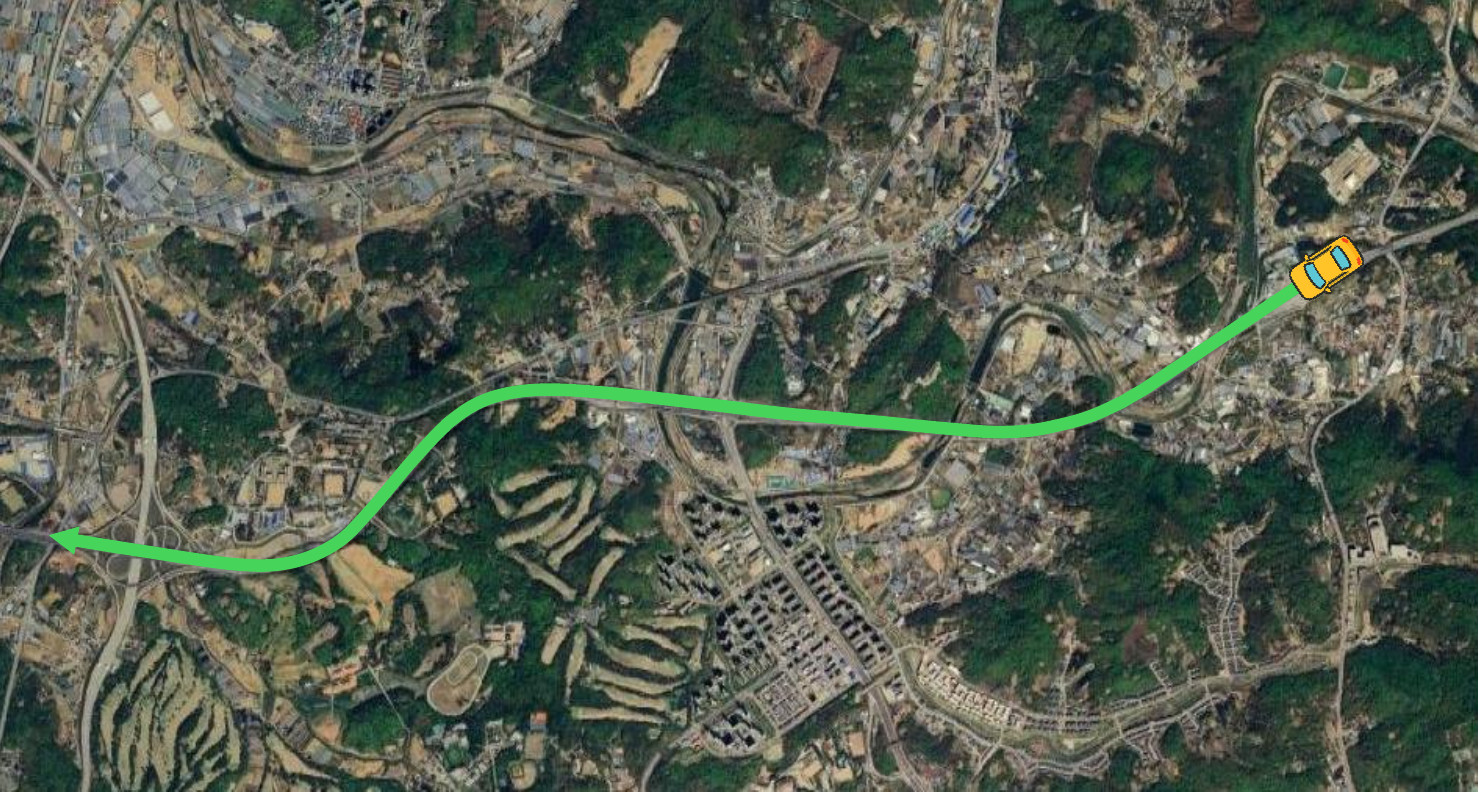}
    \caption{Vehicle used in the experiment.}
    \label{fig:road}
\end{figure}

Lane information was collected using a camera mounted on the vehicle's windshield, and this data was acquired in real-time by a PC via CAN communication at a sampling interval of 0.05 seconds. Additionally, the $\delta_d$ and $V_x$ were recorded from each sensor on the same PC with a sampling interval of 0.01 seconds. The data acquisition system setup is summarized in Fig.~\ref{fig:data_logging}.  

\begin{figure}[!ht]
    \centering
    \includegraphics[width=0.45\textwidth]{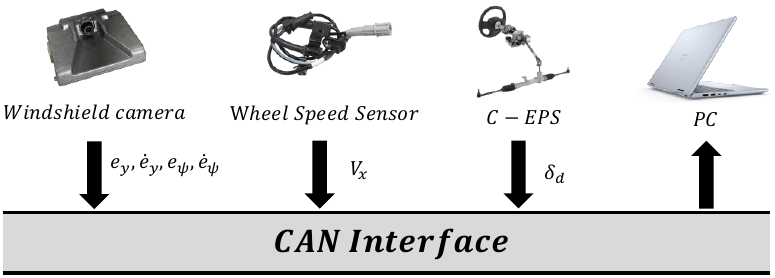}
    \caption{Vehicle data collection and recording architecture.}
    \label{fig:data_logging}
\end{figure}

Additionally, continuous driving data spanning 5 minutes, including straight and gentle curve segments, were gathered and used for training and validating the NeuroDOB.

\subsection{Experiment result}
The experiment was conducted using actual vehicle data collected during real driving conditions. The analysis compared the steering wheel angle generated by a pure LQR controller, a NeuroDOB compensated controller, and the steering wheel angle inputted by the actual driver. The results of this comparison are depicted in Fig.~\ref{fig:niro steering input}.

\begin{figure}[!ht]
    \centering
    \includegraphics[width=0.45\textwidth]{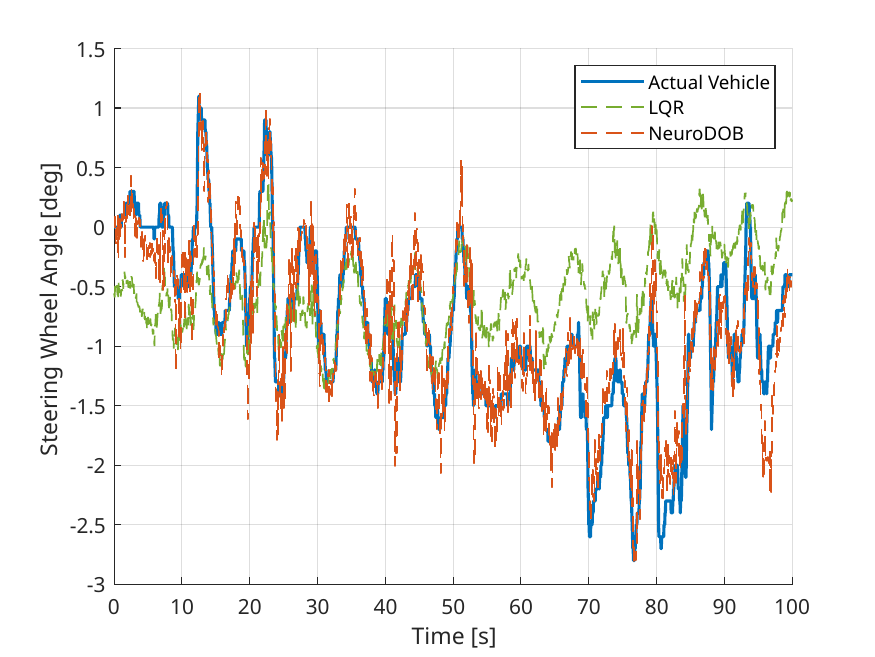}
    \caption{Steering input comparison between the LQR controller and NeuroDOB-compensated controller on actual vehicle environment.}
    \label{fig:niro steering input}
\end{figure}

Fig.~\ref{fig:niro steering input} compares steering inputs from the human driver, the pure LQR controller, and the NeuroDOB-enhanced LQR controller during a highway drive. Since most of the route consists of straight and gently curved sections, overall steering angles remain small. Between 15s and 50s, the pure LQR tracks the driver’s input reasonably well, but exhibits noticeable lag and overcompensation in the opposite direction during the initial (0–15s) and final (50–100s) segments. In contrast, the NeuroDOB-enhanced LQR consistently follows the driver’s steering across the entire 100s interval, aligning almost perfectly during curves without deviation. These results demonstrate that NeuroDOB effectively learns and reproduces driver-specific steering patterns in real time. 
Additionally, to quantitatively validate NeuroDOB’s performance, the RMSE of the steering inputs for the pure LQR controller, and the NeuroDOB-enhanced LQR controller was compared, and the results are presented in Table~\ref{tab:niro_rmse_comparison}.

\begin{table}[thb]
\caption{Actual Vehicle Experiment RMSE Comparison}
\centering
\captionsetup{justification=centering}
\small
\begin{tabular}{|l|c|c|c|}
\hline
\textbf{} & \textbf{LQR} & \textbf{NeuroDOB} & \textbf{Change} \\
\hline
$\delta_{\mathrm{d}}$ & $0.7817$ & $0.2518$ & $67.79\% \Downarrow$ \\
\hline
\end{tabular}
\label{tab:niro_rmse_comparison}
\end{table}

Table~\ref{tab:niro_rmse_comparison} shows that the NeuroDOB-enhanced LQR controller reduced the steering input RMSE from 0.7817 to 0.2518, a 67.79\% decrease compared to the pure LQR. This significant error reduction indicates that NeuroDOB compensation substantially improves control performance by accurately reproducing the driver’s fine steering intent. The higher RMSE of the pure LQR reflects its inherent model mismatch and inability to capture individual driver characteristics, resulting in persistent errors. In contrast, NeuroDOB compensates for these unmodeled dynamics and driver-specific differences in real time, enabling more precise lateral control. The dramatic decrease in RMSE confirms that NeuroDOB reduces both phase delay and amplitude error, thus dramatically enhancing path-tracking performance in real-vehicle conditions.

\section{Conclusion}
This study proposed an LQR control framework integrated with a NeuroDOB to enhance lateral vehicle trajectory tracking performance. NeuroDOB effectively compensates for unmodeled vehicle dynamics and driver-specific characteristics, significantly improving tracking accuracy compared to the baseline LQR controller.

Experimental results on two different road environments demonstrated that the application of NeuroDOB yielded more than an 86\% reduction in the RMSE of $e_y$ while maintaining yaw stability. Notably, on the unseen Road Map 2, which features increased curvature variability and more complex road segments, NeuroDOB showed improved steering control compared to pure model-based control.

However, limitations due to domain shifts between the training and testing road geometries were revealed through curvature histogram analysis. The inability of a single NeuroDOB model trained on Road Map 1 to perfectly replicate the control performance on Road Map 2 is attributed to the differences in curvature distributions between the two roads. This challenge motivates the need for enhanced diversity in training data and the adoption of modular learning approaches such as Sparse Identification of Nonlinear Dynamical Systems (SINDy), which allows for localized and partial learning.

Additionally, actual vehicle experiments were conducted using a Kia Niro platform in highway driving conditions at speeds of 90–100 km/h. The NeuroDOB enhanced LQR controller reduced the steering RMSE by 67.79\% compared to the pure LQR controller and significantly improved lane-keeping performance during driving. This validates the practical applicability of NeuroDOB in real driving environments and its capability for personalized steering pattern replication for individual drivers.

Future work will focus on expanding the training dataset to encompass a wider variety of road geometries and exploring hybrid model architectures based on SINDy to improve controller robustness and adaptability. Additionally, real-vehicle experiments based on actual driving data will be conducted to verify the practical applicability of the proposed method. This research contributes to the development of data-driven, personalized, and reliable lateral control solutions applicable to complex and diverse real-world driving scenarios. 

Conceptually, NeuroDOB demonstrates a control-level embodiment of dual-system intelligence. The coexistence of a fast, analytic System 1 (LQR) and a slow, adaptive System 2 (NeuroDOB) provides both robustness and personalization, suggesting a pathway toward cognitive-inspired vehicle controllers.

\vfill


\begin{thebibliography}{1}
\bibliographystyle{IEEEtran}

\bibitem{ref1}
Ziebinski, Adam, et al, "Review of advanced driver assistance systems (ADAS)." \textit{AIP Conference Proceedings.} Vol. 1906. No. 1. AIP Publishing LLC, 2017.

\bibitem{ref2}
Kebbati, Yassine, et al, "Lateral control for autonomous wheeled vehicles: A technical review."{\it{Asian Journal of Control}} 25.4 (2023): 2539-2563.

\bibitem{ref3}
Hwangbo, Hun, et al, "Robust sliding mode control with optimal path following for lateral motion of autonomous bus." 11 (2023): 71981-71993.

\bibitem{ref4}
Buschek, Harald, and Anthony J. Calise, "Uncertainty modeling and fixed-order controller design for a hypersonic vehicle model." \textit{Journal of Guidance, Control, and Dynamics} 20.1 (1997): 42-48.

\bibitem{ref5}Carvalho, Ashwin, et al, "Stochastic predictive control of autonomous vehicles in uncertain environments." \textit{12th international symposium on advanced vehicle control.} Vol. 9. 2014.

\bibitem{ref6}
Yang, Shaopu, Yongjie Lu, and Shaohua Li, "An overview on vehicle dynamics." \textit{International Journal of Dynamics and Control} 1.4 (2013): 385-395.

\bibitem{ref7}
Soubou, Hiromu, Takashi Ohhira, and Akira Shimada, "An Attitude Control Considering Sideslip Phenomenon Using DOB-Based LMPC-LQR for Four-Wheel Steering Vehicle." \textit{2018 57th Annual Conference of the Society of Instrument and Control Engineers of Japan (SICE).} IEEE, 2018.

\bibitem{ref8}
Soubou, Hiromu, Takashi Ohhira, and Akira Shimada, "An Attitude Control Considering Sideslip Phenomenon Using DOB-Based LMPC-LQR for Four-Wheel Steering Vehicle." \textit{2018 57th Annual Conference of the Society of Instrument and Control Engineers of Japan (SICE).} IEEE, 2018.

\bibitem{ref9}
Bünte, Tilman, et al, "Robust vehicle steering control design based on the disturbance observer." \textit{Annual reviews in control} 26.1 (2002): 139-149.

\bibitem{ref10}
True, Hans, "On the theory of nonlinear dynamics and its applications in vehicle systems dynamics." \textit{Vehicle System Dynamics} 31.5-6 (1999): 393-421.

\bibitem{ref11}
Andrzejewski, Ryszard, and Jan Awrejcewicz, \textit{Nonlinear dynamics of a wheeled vehicle.} Boston, MA: Springer US, 2005.

\bibitem{ref12}
Olson, B. J., S. W. Shaw, and G. Stépán, "Nonlinear dynamics of vehicle traction." \textit{Vehicle System Dynamics} 40.6 (2003): 377-399.

\bibitem{ref13}
Lin, Na, et al, "An overview on study of identification of driver behavior characteristics for automotive control." \textit{Mathematical Problems in Engineering} 2014.1 (2014): 569109.

\bibitem{ref14}
Liao, Xishun, et al, "A review of personalization in driving behavior: Dataset, modeling, and validation." \textit{IEEE Transactions on Intelligent Vehicles} (2024).

\bibitem{ref15}
Yi, Dewei, et al, "Implicit personalization in driving assistance: State-of-the-art and open issues." \textit{IEEE Transactions on Intelligent Vehicles} 5.3 (2019): 397-413.

\bibitem{ref16}
Kuutti, Sampo, et al, "A survey of deep learning applications to autonomous vehicle control." \textit{IEEE Transactions on Intelligent Transportation Systems} 22.2 (2020): 712-733.

\bibitem{ref17}
Kang, Chang Mook, Seung-Hi Lee, and Chung Choo Chung, "Comparative evaluation of dynamic and kinematic vehicle models." \textit{53rd IEEE Conference on Decision and Control.} IEEE, 2014.

\bibitem{ref18}
Rajamani, Rajesh, "Lateral vehicle dynamics." \textit{Vehicle Dynamics and control.} Boston, MA: Springer US, 2011. 15-46.

\bibitem{ref19}
Li, Shihua, Haibo Du, and Xinghuo Yu, "Discrete-time terminal sliding mode control systems based on Euler's discretization." \textit{IEEE Transactions on Automatic Control} 59.2 (2013): 546-552.


\bibitem{ref20}
Yang, Tao, et al, "Intelligent vehicle lateral control method based on feedforward+ predictive LQR algorithm." \textit{Actuators.} Vol. 10. No. 9. MDPI, 2021.

\bibitem{ref21}
Zheng, Zhu-an, Zimo Ye, and Xiangyu Zheng, "Intelligent vehicle lateral control strategy research based on feedforward+ predictive LQR algorithm with GA optimisation and PID compensation." \textit{Scientific Reports} 14.1 (2024): 22317.

\bibitem{ref22}
Hsu, K-Y., H-Y. Li, and Demetri Psaltis, "Holographic implementation of a fully connected neural network." \textit{Proceedings of the IEEE} 78.10 (1990): 1637-1645.

\bibitem{ref23}
Goodfellow, Ian, Yoshua Bengio, and Aaron Courville. "Deep feedforward networks." \textit{Deep learning} 1 (2016): 161-217.

\bibitem{ref24}
Goodfellow, Ian, et al, \textit{Deep learning.} Vol. 1. No. 2. Cambridge: MIT press, 2016.

\bibitem{ref25}
Glorot, Xavier, and Yoshua Bengio, "Understanding the difficulty of training deep feedforward neural networks." \textit{Proceedings of the thirteenth international conference on artificial intelligence and statistics.} JMLR Workshop and Conference Proceedings, 2010.

\bibitem{ref26}
Hu, Tingshu, and Zongli Lin, "Composite quadratic Lyapunov functions for constrained control systems." \textit{IEEE Transactions on Automatic Control} 48.3 (2003): 440-450.

\bibitem{ref27}
De Oliveira, Maurício C., Jacques Bernussou, and José C. Geromel, "A new discrete-time robust stability condition." \textit{Systems \& control letters} 37.4 (1999): 261-265.

\bibitem{ref28}
Wang, Qi, et al, "A comprehensive survey of loss functions in machine learning." \textit{Annals of Data Science} 9.2 (2022): 187-212.

\bibitem{ref29}
Su, Jinya, Wen-Hua Chen, and Jun Yang, "On relationship between time-domain and frequency-domain disturbance observers and its applications." \textit{Journal of Dynamic Systems, Measurement, and Control} 138.9 (2016): 091013.

\bibitem{ref30}
Kim, Kyung-Soo, Keun-Ho Rew, and Soohyun Kim. "Disturbance observer for estimating higher order disturbances in time series expansion." \textit{IEEE Transactions on automatic control} 55.8 (2010): 1905-1911.

\end{thebibliography}
\end{document}